\documentclass[10pt, conference, letterpaper]{IEEEtran}

\usepackage{framed}
\usepackage{hyperref}
\usepackage{tcolorbox}
\usepackage{cite}
\usepackage{amsmath}
\usepackage{amsfonts}
\usepackage{mathtools}
\usepackage{amssymb}
\usepackage{amsthm}
\usepackage[ruled,linesnumbered,noend]{algorithm2e}
\usepackage{cases}
\usepackage{graphicx}
\usepackage{bbm}
\usepackage{subcaption}
\usepackage{svg}

\newtheorem{theorem}{{\bf Theorem}}
\newtheorem{conjecture}{{\bf Conjecture}}
\newtheorem{lemma}{{\bf Lemma}}

\DeclareMathOperator*{\argmax}{arg\,max}

\IEEEoverridecommandlockouts

\begin{document}
%
\title{Utility Maximization in Wireless Backhaul Networks with Service Guarantees}
\author{Nicholas Jones and Eytan Modiano
\thanks{The authors are with the Laboratory for Information and Decision Systems (LIDS), Massachusetts Institute of Technology, Cambridge, MA 02139. \ 
email: \{jonesn, modiano\}@mit.edu}}

\IEEEaftertitletext{\vspace{-0.6\baselineskip}}
\maketitle

\begin{abstract}
    We consider the problem of maximizing utility in wireless backhaul networks, where utility is a function of satisfied service level agreements (SLAs), defined in terms of end-to-end packet delays and instantaneous throughput. We model backhaul networks as a tree topology and show that SLAs can be satisfied by constructing link schedules with bounded inter-scheduling times, an NP-complete problem known as pinwheel scheduling. For symmetric tree topologies, we show that simple round-robin schedules can be optimal under certain conditions. In the general case, we develop a mixed-integer program that optimizes over the set of admission decisions and pinwheel schedules. We develop a novel pinwheel scheduling algorithm, which significantly expands the set of schedules that can be found in polynomial time over the state of the art. Using conditions from this algorithm, we develop a scalable, distributed approach to solve the utility-maximization problem, with complexity that is linear in the depth of the tree.
\end{abstract}

\maketitle

\section{Introduction}

Next-generation wireless networks must support increasing levels of real-time control and inference as these technologies become ubiquitous. In many such systems, the controller or learning model is located at a centralized node or edge server, and the system relies on network communication to transport data from distributed devices to the server in real time. Best-effort service is often insufficient for these systems, which demand hard guarantees on throughput and latency to operate safely and reliably.

In addition to hard service guarantees, these systems require ever-increasing data rates, which necessitate an ever-increasing power budget. Power is already a significant operating expense in 5G, and low-power technologies are a main focus of 6G development as consumption is expected to rise even further. At the same time, the cost of hardware deployment is rapidly dropping, incentivizing infrastructure providers to install more and more access points (APs) and small cells with wireless backhaul links. This creates a natural tree topology, with devices connected to APs at the leaves, and the server located at the root.

In this work we consider networks of this topology and uplink traffic with hard service requirements, dictated by service-level agremeents (SLAs) between the network operator (NO) and customers. We focus on uplink traffic because the data being sent to the model is typically much larger than the model output being returned to the device, but the same analysis holds for downlink traffic. We assume that the NO receives some utility (revenue) from satisfying each SLA, which is a function of the service requirements it must satisfy, and we derive admission control and offline scheduling policies to maximize the NO's utility, while guaranteeing each SLA is satisfied. In particular, our policies are able to make hard guarantees on throughput and end-to-end packet delays. While we assume a wireless topology subject to interference, our results can be generalized to any TDMA backhaul network.

There is a considerable body of work on network utility maximization. The problem was first formulated by Kelly in~\cite{kelly1998rate}, who proposed allocating rates to users to maximize utility in the network. A distributed solution to this problem was developed in~\cite{low1999optimization}. It has been studied extensively in the context of stochastic network optimization, using Lyapunov-based approaches to maximize utility while ensuring network stability~\cite{neely2010stochastic}. A summary of cross-layer techniques for wireless networks can be found in~\cite{georgiadis2006resource}.

There has also been work on utility maximization with service guarantees. Hou and Kumar introduced a framework for meeting service requirements in single-hop wireless networks~\cite{hou_theory_2009} and showed that this can be used to maximize utility in~\cite{hou2010utility}, where utility is determined by the fraction of user traffic which is served before its deadline. This framework was extended to multi-hop networks in~\cite{li2012scheduling}.

Additionally, there is a growing body of work on scheduling in wireless backhaul networks. In~\cite{kabbani_distributed_2007}, a throughput-maximizing scheduling policy is developed for wireless backhaul, and in~\cite{gopalam2022distributed} a backpressure based policy is proposed with the same objective. In~\cite{narlikar2010designing}, a scheduling policy is designed which achieves at most twice the delay of an equivalent wired backhaul by multiplexing over both time and frequency. An online learning algorithm is used in~\cite{han2018online} to minimize average wireless backhaul delay, while a cross-layer routing, scheduling, and resource allocation design is explored in~\cite{cao2007multi}. 

Offline scheduling for throughput maximization in general wireless networks has also received considerable attention~\cite{jain2003impact,hajek_link_1988,kodialam2003characterizing}. When deadline constraints are considered, optimal schedules can be found by relaxing interference and capacity constraints~\cite{singh2018throughput,singh2021adaptive}. Without these relaxations, the authors of~\cite{djukic2007quality,djukic_delay_2009} show that packet delays can be bounded as a function of the schedule length by solving a mixed-integer program or using heuristic approaches.

Recently, a scheduling framework was proposed to make hard throughput and delay guarantees in wireless networks with interference constraints and general topologies~\cite{jones_optimal_2024}. Using a network slicing model, the authors show that tight deadline guarantees can be made by bounding inter-scheduling times for each link on a packet's route. Constructing schedules with bounded inter-scheduling times is known as pinwheel scheduling and was first introduced in~\cite{holte1989pinwheel}. It is known to be NP-complete~\cite{mok_algorithms_1989}, but polynomial time algorithms exist which can form schedules under certain sets of conditions~\cite{chan_general_1992,chan_schedulers_1993,han_distance_1996,lin_pinwheel_1997}. 

In this work, we build on the network slicing model and inter-scheduling conditions of~\cite{jones_optimal_2024} to make simultaneous throughput and delay guarantees for flows in a wireless backhaul network. We introduce a novel state of the art algorithm called Inductive Scheduling ($IS$) that significantly expands the set of pinwheel schedules that can be found in polynomial time. We derive feasibility conditions for pinwheel schedules based on this algorithm, and optimize over these scheduling conditions and flow admission decisions to efficiently maximize utility for a batch of flows, while guaranteeing that each flow's SLA is satisfied.

The remainder of this paper is organized as follows. In Section~\ref{sec:sysmodel}, we introduce the system model, summarize the existing results on which our work is based, and formally define our objective. In Section~\ref{sec:symmetricarrivals}, we show that when the tree is perfectly symmetric, simple round-robin schedules can be optimal, but a more complex policy based on pinwheel scheduling is required when asymmetry exists. In Section~\ref{sec:inductivescheduling}, we introduce the $IS$ algorithm and show that it constructs schedules in polynomial time which were previously unattainable. In Section~\ref{sec:generalarrivals}, we show that conditions from the $IS$ algorithm can be used to solve our utility maximization problem, formulated as a mixed-integer program, and we present a scalable distributed algorithm for doing so. In Section~\ref{sec:simulations}, we provide numerical results for the $IS$ algorithm and our utility maximization framework.

\section{System Model}\label{sec:sysmodel}

We consider a wireless network with a tree topology, denoted by $G=(V,E)$, where $V$ is the set of nodes (vertices on the graph), and $E$ is the set of links (edges on the graph). The tree consists of $D$ levels, with the server (root node) at level $0$ and the leaves (APs) at level $D-1$. Denote the set of level $d$ nodes as $V_d$, and the level of node $v$ as $d(v)$. Each node $v$ has $N_v$ children, and denote the set of its children as $\mathcal{C}_v$, the set of all its descendants as $V_v^-$, and the set of its ancestors as $V_v^+$. We consider uplink traffic so each node has one outgoing directional link connecting it to its parent. Because it is unique, we denote the outgoing link from node $v$ as $e(v)$.

Time is slotted, with the duration of one time slot equal to the transmission time of a single packet over a link of unit capacity. Physically, a time slot can be thought of as the smallest schedulable interval as dictated by the physical and link layers, and a unit of capacity as the capacity needed to send one packet over this time interval. Each link $e \in E$ has a fixed capacity $c_e$, which can vary between links. Links are error-free but subject to local interference, where each set of children that share a parent node mutually interfere, so only one child can be scheduled in each slot. Denote the set of links scheduled at time $t$ as $\mu(t)$, and let $\mu_e(t) = 1$ if link $e$ is scheduled at time $t$ and $0$ otherwise. Denote the maximum inter-scheduling time of link $e$, i.e., the largest continuous time interval between slots where link $e$ is scheduled, as $k_e$. Let $\boldsymbol{k}$ be the vector of $k_e$ values for all $e \in E$, and $\boldsymbol{k_v}$ be the vector of $k_{e(v')}$, for all $v' \in \mathcal{C}_v$. We use a superscript $\pi$ to denote these quantites under a given scheduling policy $\pi$.

Network traffic takes the form of flows, where each flow corresponds to a customer connected wirelessly to an AP at level $D-1$. While connected, these devices form level $D$ of the tree. We assume batch arrivals, where a batch of flows arrives simultaneously, and the network makes admission decisions and computes an offline schedule for the entire batch. Before establishing a connection, each device establishes an SLA with the NO, dictating the level of service it requires. These SLAs are defined in terms of a maximum instantaneous arrival rate $\lambda$, a packet deadline $\tau$, and an SLA duration $T$. For ease of exposition, we assume in this work that SLAs are equivalent for all flows. This can easily be extended to a flow-specific $\lambda$ and $\tau$ values, which we will address in a future version of this work.

An SLA initialized at time $t_0$ dictates that if $\lambda_i(t) \leq \lambda$ packets belonging to $f_i$ are generated at time $t$, each packet must be delivered to the server by time $t+\tau$, for every $t_0 \leq t < t_0 + T$. We assume the NO cannot accept an SLA unless it is guaranteed to be met. When a flow's SLA is satisfied, the NO receives some utility as a function of $\lambda$ and $\tau$. Because SLAs are symmetric, maximizing utility is equivalent to maximizing the number of supported SLAs, which we denote by $\sigma$. Note that this framework can also support flows which require guaranteed rates without hard deadlines, by setting $\tau$ to be large.

Denote flow $f_i$'s route, from the customer at level $D$ to the server at the root, as $\mathcal{T}^{(i)}$, and the set of all leaf-to-root paths as $\mathcal{T}$. Each link $e \in \mathcal{T}^{(i)}$ is assigned a slice of capacity, which is reserved specifically for flow $f_i$. We denote the amount of capacity reserved for $f_i$, or the slice width, as $w_{i,e}$. In addition, each slice has its own dedicated queue, which decouples the queueing of different flows. We assume all scheduling policies are work-conserving, so when a link $e$ is scheduled, it serves the smaller of the current queue size and the slice width of each flow.

\subsection{Preliminary Results}

In~\cite{jones_optimal_2024}, the authors show that for any wireless topology and interference constraints, and using the same network slicing model described above, deadline guarantees can be made by satisfying conditions on maximum inter-scheduling times $\boldsymbol{k}$. Specifically, they show the following result for a given set of flows $\mathcal{F}$.
\begin{lemma}[\!\cite{jones_optimal_2024}]\label{lemma:suffconds}
    If a scheduling policy exists which satisfies SLAs for all flows in $\mathcal{F}$, then there must exist an SLA-satisfying cyclic policy $\pi$ with period $K^{\pi}$, such that 
    \begin{equation}
        \mu^{\pi}(t) = \mu^{\pi}(t+K^{\pi}), \ \forall \ t_0 \leq t < T - K^{\pi}.
    \end{equation}
    Furthermore, if slice widths are set such that
    \begin{equation}\label{eq:slicewidthcond}
        w_{i,e} = \lambda k_e^{\pi}, \ \forall \ f_i \in \mathcal{F}, e \in \mathcal{T}^{(i)},
    \end{equation}
    then SLAs are satisfied for all flows $\mathcal{F}$ if 
    \begin{align}
        \sum_{e \in \mathcal{T}^{(i)}} k_e^{\pi} \leq \tau, \ \forall \ f_i \in \ &\mathcal{F},\label{eq:deadlinecond} \\ 
        \sum_{f_i \in \mathcal{F} : e \in \mathcal{T}^{(i)}} \lambda k_e^{\pi} \leq c_e, \ \forall \ &e \in E.\label{eq:slicecapcond}
    \end{align}
\end{lemma}

This result shows that we can restrict ourselves to cyclic policies without loss of optimality, and shows how to construct an SLA-satisfying policy in terms of slice widths and inter-scheduling times. The deadline constraint~\eqref{eq:deadlinecond} combined with the slice width condition~\eqref{eq:slicewidthcond} guarantees that all deadlines are met by forcing links to be scheduled somewhat regularly and queue sizes to remain small, while the slice capacity condition~\eqref{eq:slicecapcond} ensures that a link has enough capacity to support the allocated slices. We encourage readers to reference~\cite{jones_optimal_2024} for more details.

Define the class of cyclic, work-conserving scheduling policies that satisfy our local interference constraints as $\Pi_c$. In order for queues to remain stable and for any finite deadline to be met as the SLA duration $T$ goes to infinity, the service rate of each queue must be at least as large as the arrival rate~\cite{jones_optimal_2024}. Defining the time-average activation rate of link $e$ under policy $\pi$ as
\begin{equation}
    \bar{\mu}_e^{\pi} \triangleq \frac{1}{K^{\pi}} \sum_{t=0}^{K^{\pi}-1} \mu_e^{\pi}(t),
\end{equation}
the queue stability condition can be written as
\begin{equation}\label{eq:queuestability}
    w_{i,e} \bar{\mu}_e^{\pi} \geq \lambda.
\end{equation}
If the inter-scheduling time of link $e$ is always the same under a policy $\pi$, i.e., it is scheduled \textit{exactly} every $k_e^{\pi}$ slots, then $k_e^{\pi} = 1/\bar{\mu}_e^{\pi}$, and the slice width condition~\eqref{eq:slicewidthcond} represents the smallest stable slice width under $\pi$. The larger $k_e$ is relative to $1/\bar{\mu}_e$, the larger the slice width must be and the larger the minimum feasible $\tau$ becomes. This shows that schedules with regular inter-scheduling times are more efficient at satisfying SLAs. We will see this in more detail in the next section.



\section{Utility Maximization on Symmetric Trees}\label{sec:symmetricarrivals}


Recall that our objective is to maximize the number of supported SLAs $\sigma$ over all admission decisions and scheduling policies. We begin by analyzing this problem under symmetry conditions. Assume the tree $G$ is perfectly symmetric at each level, such that each node $v$ at level $d$ has $N_{d+1}$ children, and $N_D$ flows arrive at each AP. Further assume that capacities are fixed across each level of the tree, with $c_{e(v)} = c_d$ for each node $v$ at level $d$. To help analyze the problem in this setting, it is instructive to consider the opposite perspective and to define the feasible set of $(\lambda, \tau)$ pairs, such that \textit{all} SLAs can be supported by some policy $\pi \in \Pi_c$ on $G$. Denote this region as $\Lambda(G)$, and define the throughput-optimal rate
\begin{equation}
    \lambda^*(G) \triangleq \max_{(\lambda, \tau) \in \Lambda(G)} \lambda
\end{equation}
as the largest value of $\lambda$ in this set independent of $\tau$. When $\lambda \leq \lambda^*(G)$, we say that $\lambda$ is rate-feasible on $G$. Similarly define the delay-minimizing deadline 
\begin{equation}
    \tau^*(G) \triangleq \min_{(\lambda,\tau) \in \Lambda(G)} \tau
\end{equation}
as the smallest value of $\tau$ in this set independent of $\lambda$. When $\tau \geq \tau^*(G)$ we say that $\tau$ is deadline-feasible on $G$. We will see that analyzing these quantities individually allows us to define $\Lambda(G)$ and provide insight into maximizing $\sigma$.

We first consider $\lambda^*(G)$, and assume without loss of generality that $\tau \geq D K^{\pi}$ for any policy $\pi$ with finite period $K^{\pi}$. Then deadline constraints~\eqref{eq:deadlinecond} are satisfied under any $\pi$ that schedules each link at least once per scheduling period, and we can drop the deadline constraints in the analysis. To ensure feasibility as the SLA duration $T \to \infty$, we need only ensure that queue stability conditions~\eqref{eq:queuestability} are satisfied and that allocated slice widths do not exceed link capacities. Setting slice widths to their smallest value while ensuring stability yields $w_{i,e(v)} = \lambda / \bar{\mu}_{e(v)}$, for each flow $f_i$ and link $e$. 

There are $|\mathcal{F}_v| = \prod_{d'=d+1}^D N_{d'}$ flows which pass through link $e(v)$, by the symmetry of $G$, so the capacity constraint at each link $e(v)$ is 
\begin{equation}
    \lambda \prod_{d'=d(v)+1}^D N_{d'} \leq \bar{\mu}_{e(v)} c_{e(v)}.
\end{equation}
Then $\lambda^*(G)$ is the solution to the LP
\begin{align}
\begin{aligned}\label{eq:lambdastaropt}
    \max_{\pi \in \Pi_c} \ &\lambda \\ 
    \text{s.t.} \ &\lambda \prod_{d'=d(v)+1}^D N_{d'} \leq \bar{\mu}_{e(v)}^{\pi} c_{e(v)}, \ \forall \ v \in V_c, \\
    &\sum_{v' \in \mathcal{C}_v} \bar{\mu}_{e(v')}^{\pi} \leq 1, \ \forall \ v \in V_c, \\ 
\end{aligned}
\end{align}
where the second constraint ensures that no more than one link per set of children is scheduled per slot.

\begin{lemma}\label{lemma:lambdastar}
    When $G$ is perfectly symmetric at each level,
    \begin{equation}\label{eq:lambdastar}
        \lambda^*(G) = \min_{0 \leq d < D} \Big(\prod_{d'=d+1}^D N_{d'} \Big)^{-1} c_{d+1},
    \end{equation}
    and this is achieved by a universal round-robin (URR) scheduling policy, defined as a round-robin policy at each set of children.
\end{lemma}

\begin{proof}
    Let $\lambda_d^*(G)$ be the throughput-optimal rate for the subtree rooted at node $v$ in level $d$ of the tree. By symmetry, this quantity is the same for all nodes at level $d$. Every flow which passes through a node at level $d$ must first pass through one of its children at level $d+1$, so $\lambda_d^*(G) \leq N_{d+1} \lambda_{d+1}^*(G)$ for any level $d$. Combining this with the constraints in~\eqref{eq:lambdastaropt}, $\lambda_d^*(G)$ is the solution to 
    \begin{align}
    \begin{aligned}
        \max_{\lambda, \pi \in \Pi_c} \ &\lambda \\ 
        \text{s.t.} \ &\lambda \prod_{d'=d+2}^D N_{d'} \leq \bar{\mu}_{e(v')}^{\pi} c_{d+1}, \ \forall \ v' \in \mathcal{C}_v, \\
        &\sum_{v' \in \mathcal{C}_v} \bar{\mu}_{e(v')}^{\pi} \leq 1, \\ 
        &\lambda \leq \lambda_{d+1}^*(G).
    \end{aligned}
    \end{align}.

    Without the last constraint, it is straightforward to see that $\lambda$ is maximized when $\bar{\mu}_{e(v')}^{\pi} = 1/N_{d+1}$ for each child $v'$. This is achieved by a round-robin policy over the set of children, which yields a value $c_{d+1} / \prod_{d'=d+1}^D N_{d'}$. The optimal value is then 
    \begin{equation}\label{eq:lambda_d}
        \lambda_d^* = \min \Big \{ c_{d+1} / \prod_{d'=d+1}^D N_{d'}, \lambda_{d+1}^* \Big \},
    \end{equation}
    and this is always achievable by a round-robin policy.

    At the level $D-1$ APs, the children represent the individual flows, and $\lambda_D^*$ is limited only by the amount of traffic each flow can send. For our purposes of throughput-maximization, we assume $\lambda_D^* \to \infty$, and from~\eqref{eq:lambda_d} we immediately see that $\lambda_{D-1}^* = c_D / N_D$. Then by induction over each level of the tree, we have 
    \begin{equation}
            \lambda^*(G) = \min_{0 \leq d < D} \Big(\prod_{d'=d+1}^D N_{d'} \Big)^{-1} c_{d+1},
    \end{equation}
    and this is achieved by a round-robin policy at each set of children in the tree.
\end{proof}

Next we turn to analyzing $\tau^*(G)$, which is equivalent to the min-max delay seen by any packet. Again because $\tau^*(G)$ is defined independently of $\lambda$, we can assume without loss of generality that $\lambda = \epsilon$, for some strictly positive but arbitrarily small $\epsilon$.
Then the queue stability and slice capacity constraints are satisfied under any policy that schedules each link at least once per scheduling period, and to ensure feasibility we need only ensure that deadline constraints are met. Recall from~\eqref{eq:deadlinecond} that a packet can meet any deadline $\tau \geq \sum_{e \in \mathcal{T}^{(i)}} k_e$. In other words, the maximum delay seen by any packet is upper bounded by the sum of inter-scheduling times along its route. Therefore,
\begin{equation}\label{eq:taugbound}
    \tau^*(G) \leq \max_{T \in \mathcal{T}} \sum_{e \in T} k_e
\end{equation}
where recall that $\mathcal{T}$ denotes the set of all leaf-to-root paths in the tree.

Each route sees the same scheduling constraints at a given level of the tree due to the symmetry of $G$. Therefore, minimizing the maximum inter-scheduling time at each level of the tree intuitively minimizes the upper bound in~\eqref{eq:taugbound}, because each route will see the same inter-scheduling times at each node. For a given node $v$, the solution to 
\begin{align}
\begin{aligned}
    \min_{\boldsymbol{k_v}} &\max_{v' \in \mathcal{C}_v} k_{e(v')} \\
    \text{s.t.} \ &\boldsymbol{k_v} \ \text{schedulable}
\end{aligned}
\end{align}
is given by a round-robin policy where $k_{e(v')} = |\mathcal{C}_v| = N_d$ for each $v'$ at level $d$. It turns out that a URR policy not only minimizes the bound in~\eqref{eq:taugbound}, but that this bound is tight.

\begin{lemma}\label{lemma:taustar}
    When $G$ is perfectly symmetric at each level,
    \begin{equation}\label{eq:taustar}
        \tau^*(G) = \sum_{d=1}^D N_d,
    \end{equation}
    and both throughput- and deadline-optimality are achieved by a URR policy. Therefore,
    \begin{equation}
        \Lambda(G) = \{ (\lambda,\tau) \ | \ \lambda \leq \lambda^*(G), \tau \geq \tau^*(G) \},
    \end{equation}
    and for any $(\lambda,\tau) \in \Lambda(G)$, a URR policy maximizes utility by supporting all SLAs on the tree $G$.
\end{lemma}

\begin{proof}
    Consider the set of $N_1$ children belonging to the root node. Because only one can be scheduled per slot, at any time $t$ there must be at least one child which has not been scheduled since $t-N_1$. Without loss of generality, pick one such child and consider its $N_2$ children. They are subject to the same interference constraints, so in the slot $t' \leq t - N_1$ where their parent was last scheduled, there must have been at least one child which had not been scheduled since $t'-N_2 \leq t - N_1 - N_2$.

    Proceeding by induction, there must be at least one leaf node at time $t$ which has not had scheduling opportunities for a packet which arrived at time $\hat{t} \leq t - \sum_{d=1}^D N_d$ to be delivered at the root. Because packets arrive in every slot, this packet is guaranteed to exist and experiences an end-to-end delay of $t-\hat{t}$. Therefore, the largest delay seen by any packet is $\tau^* \geq t - \hat{t} \geq \sum_{d=1}^D N_d$.

    Recall from~\eqref{eq:deadlinecond} that any deadline $\tau$ is achievable when $\sum_{e \in \mathcal{T}^{(i)}} k_e \leq \tau$ for every flow $f_i$. Under a URR policy, $k_e = N_d$ for a link $e$ at level $d$, which yields an achievable deadline $\sum_{1 \leq d \leq D} N_d$. This is equal to the lower bound above, so this lower bound is both achievable and met by a URR policy.
\end{proof}

This lemma provides an exact characterization of $\Lambda(G)$ and shows that both $\lambda^*(G)$ and $\tau^*(G)$ can be jointly achieved by the same simple and intuitive URR policy. For $(\lambda, \tau)$ that lie within $\Lambda(G)$, this policy supports all flows and is by definition utility maximizing. This result also shows that, perhaps suprisingly, increasing the number of hops in a backhaul network actually decreases packet delay. If each flow was connected to a single AP, then the minimum achievable deadline would be the total number of flows $\prod_d N_d$. By introducing additional hops, the scheduling policy is able to mulitplex over both time and space, thereby decreasing $\tau^*$.

\subsection{Greedy Pruning}

We next consider the case when $\tau < \tau^*(G)$, and the network cannot simultaneously satisfy deadline constraints for all flows. We will construct a subset of supported flows by pruning branches from the tree and cutting off all flows connected to these branches. When a leaf is pruned, this corresponds to cutting off a single flow, and when a branch farther up the tree is pruned, this corresponds to cutting off all flows whose APs belong to the subtree rooted at that branch node.

For a given tree $G$ and deadline $\tau$, let $\bar{\tau}(G) = \tau^*(G) - \tau$, which is the amount of time needed to be ``cut'' from $\tau^*(G)$ to achieve $\tau$. Then consider the set of nodes $V_{d-1}$ at level $d-1$. If a branch is pruned from each of these nodes, the value of $N_d$ is decremented by one, and we refer to this as a level $d$ prune. Note that we are pruning symmetrically across a level, so the resulting tree $G'$ is still symmetric and Lemmas~\ref{lemma:lambdastar} and~\ref{lemma:taustar} still hold on $G'$. Denote the new value of $N_d$ after pruning as $N_d(G')$. Then
\begin{equation}
    \tau^*(G') = \sum_{d'=1}^D N_{d'}(G') = \sum_{d'=1}^D N_{d'} - 1 = \tau^*(G) - 1,
\end{equation}
so performing $\bar{\tau}(G)$ of these pruning operations yields a tree $G^*$ with $\tau^*(G^*) = \tau$. Note that $\tau^*(G^*)$ is independent of which levels are pruned.

Each level $d$ prune cuts off $1/N_d$ of the flows connected to each node in $V_{d-1}$, so by symmetry it cuts off $1/N_d$ of all flows in the tree. Pruning at any level has the same effect on deadlines, but clearly pruning at level $d$ as opposed to level $d'$ cuts off fewer flows when $d > d'$. Based on this insight, we define a Greedy Pruning algorithm which performs $\bar{\tau}(G)$ sequential pruning operations, each time selecting a node to prune at any level $d \in \argmax_d N_d(G')$, where $G'$ is the current state of the tree at that operation. This algorithm is guaranteed to reach deadline feasibility while minimizing the number of flows which are pruned, producing the largest deadline-feasible subset of $\mathcal{F}$.

\begin{theorem}\label{th:greedyprune}
    Let $G$ be perfectly symmetric at each level, and $G^*$ be the output of the Greedy Pruning algorithm on $G$. If $\lambda \leq \lambda^*(G^*)$, then $(\lambda,\tau) \in \Lambda(G^*)$, and a URR policy on $G^*$ is utility-maximizing on $G$.
\end{theorem}

\begin{proof}
    Each prune of a node $v$ branch decrements $\bar{\tau}$ for all leaves in the subtree rooted at $v$, so for a given tree $G$, we must perform $\bar{\tau}(G)$ of these operations. As discussed above, pruning branches at level $d$ cuts off $1/N_d$ of the flows in the tree. Consider a set of levels $P$ which are pruned. Then the remaining number of flows after pruning is complete is given by 
    \begin{equation}
        \prod_{d \in P} \frac{N_d-1}{N_d},
    \end{equation}
    where it is understood that if level $d$ is pruned more than once, $N_d$ is decremented after each prune. Each term in the product is maximized when $N_d$ is maximized, so greedily choosing $d \in \argmax N_d$ must be optimal at each step. If not, then there must exist some $d' \notin P$ which can replace some element in $P$ and decrease the product. By definition, however, $N_{d'} \leq N_d$ for all $d \in P$, so this is a contradiction.

    The algorithm continues until the augmented tree $G^*$ is deadline-feasible, so if $\lambda \leq \lambda^*(G^*)$, then $(\lambda, \tau) \in \Lambda(G^*)$ by definition. Because $G$ started as a symmetric tree and Greedy Pruning performs symmetric operations, $G^*$ is also symmetric, and by Lemma~\ref{lemma:taustar}, $URR$ is utility-maximizing on $G^*$. Finally, because the minimum number of flows have been pruned from the tree to reach a deadline-feasible $G^*$, this must also be utility-maximizing for $G$.
\end{proof}

From~\eqref{eq:lambdastar}, we observe that $\lambda^*$ can only increase as branches are pruned, so $\lambda^*(G^*) \geq \lambda^*(G)$, and even if $\lambda$ is rate-infeasible on the original tree $G$, it may be rate-feasible on the Greedy Pruning output $G^*$. Unfortunately, if $\lambda$ is not rate-feasible on $G^*$, there is no Greedy Pruning alternative that yields a utility-maximizing policy, and there may not exist a URR policy that is utility maximizing.

\begin{figure}
    \centering
    \includegraphics[width=0.35\textwidth]{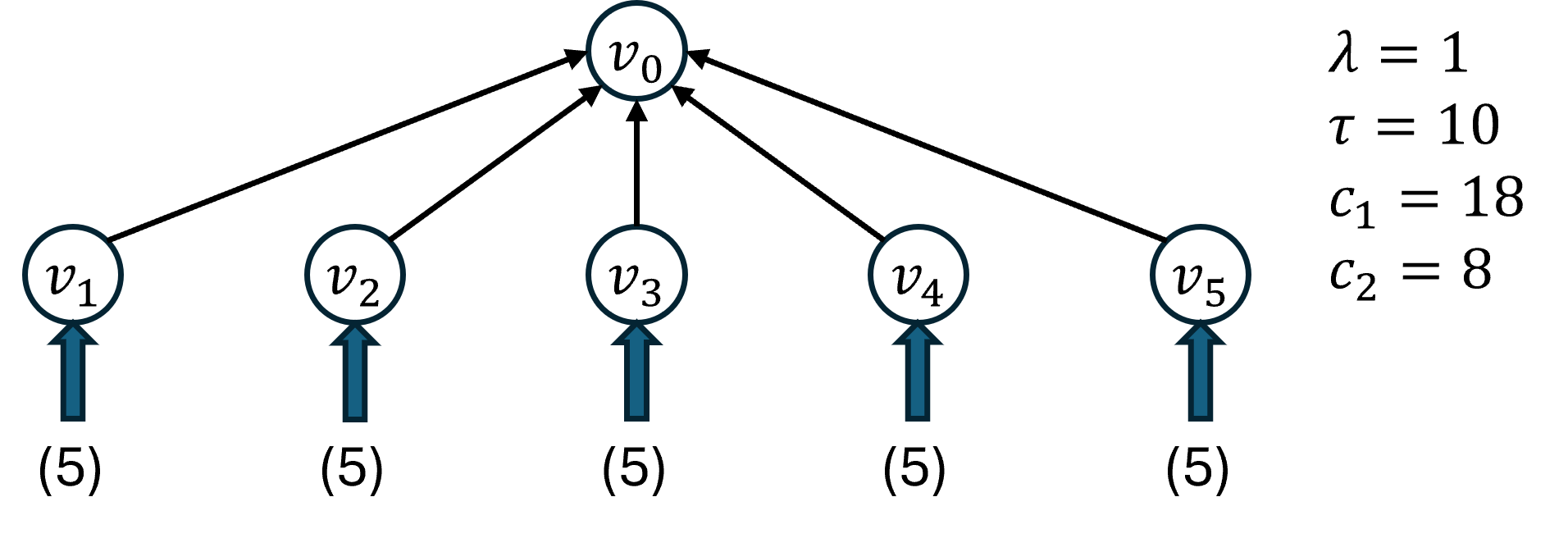}
    \caption{Example tree with $D=2$ and $N_D=8$}
    \label{fig:example-tree}
\end{figure}

To see this, consider the depth-two tree $G$ in Figure~\ref{fig:example-tree}. From~\eqref{eq:taustar}, we can immediately see that $\tau^*(G) = 10 \leq \tau$, so $\tau = 10$ is deadline-feasible on $G$. Next we use~\eqref{eq:lambdastar} to compute $\lambda^*(G) = 18/25 > \lambda$, so $\lambda=1$ is rate-infeasible on $G$. One can easily verify that pruning once at either level of the tree is insufficient to make $\lambda$ rate-feasible, so we perform two pruning operations. Pruning two branches at either level yields $\lambda^*(G^*) = 18/15$ and $\sigma = 15$, while pruning one branch at each level yields $\lambda^*(G^*) = 18/16$ and $\sigma=16$, making it the optimal URR policy. We show the pruned tree and this policy on the left of Figure~\ref{fig:example-sol}.

On the right we show the optimal policy for the tree $G$ and note that it is not a URR policy. By pruning asymmetrically and deviating from URR, we are able to support an additional flow. In particular, let the flows at each AP be scheduled round-robin and the level-$1$ APs follow the schedule 
\begin{equation}
    \pi = \{ v_1, v_3, v_2, v_4, v_1, v_5, v_2, v_3, v_1, v_4, v_2, v_5 \},
\end{equation}
which has a length of $12$ slots and repeats itself indefinitely. One can verify that this satisfies the vector of inter-scheduling times in Figure~\ref{fig:example-sol}, and that both slice capacity and deadline constraints are met.

This example shows the need for a more complete framework to maximize utility under general conditions. In the subsequent sections we will develop this framework, which can handle not only rate-infeasibililty on symmetric trees, but also asymmetric topologies and arrivals. For ease of exposition we will continue to focus on the case where flows have symmetric parameters $\lambda$ and $\tau$, so our objective is still to maximize $\sigma$. This can easily be extended to flows with general parameters, and due to space constraints we will include this extension in a future version of this work.

Let $\sigma_v^{\pi}$ be the number of supported flows in the subtree rooted at node $v$ under policy $\pi$. Using the conditions in Lemma~\ref{lemma:suffconds}, the general utility maximization problem can be written as the following integer program, where $\tilde{\mathcal{T}}$ represents the set of leaf-to-root paths which are not pruned by the admission decisions.
\begin{subequations}
\begin{align}\label{eq:globalopt}
    (P1) \ : \ \max_{\pi \in \Pi_c, \tilde{\mathcal{T}} \subseteq \mathcal{T}} \ &\sigma^{\pi} \\
    \text{s.t.} \ &\sigma_v^{\pi} = \sum_{v' \in \mathcal{C}_v} \sigma_{v'}^{\pi}, \ \forall \ v \in V, \\ 
    &\sigma_v^{\pi} \leq N_v, \ \forall \ v \in V_{D-1}, \\ 
    &\sigma_v^{\pi} \lambda k_{e(v)}^{\pi} \leq c_{e(v)}, \ \forall \ v \in V,\label{subeq:slicecap} \\ 
    &\sum_{e \in T} k_e^{\pi} \leq \tau, \ \forall \ T \in \tilde{\mathcal{T}}, \\ 
    &\boldsymbol{k}_v^{\pi} \ \text{schedulable}, \ \forall \ v \in V, \\
    &\sigma_v^{\pi}, k_{e(v)}^{\pi} \in \mathbb{Z}_+, \ \forall \ v \in V.
\end{align}
\end{subequations}
The first constraint is a flow conservation condition, the second ensures that we can support no more flows than are present at each AP, the third is the slice capacity constraint, and the fourth is the deadline constraint.

Before addressing how to solve this program, we first must explore the schedulability constraint for each set of inter-scheduling times $\boldsymbol{k}_v$. This is the topic of the next section, which introduces an iterative solution to the pinwheel scheduling problem. This section is not essential to understanding the rest of the paper, and can be skipped if the reader is not interested in the details.

\begin{figure}
    \centering
    \includegraphics[width=0.48\textwidth]{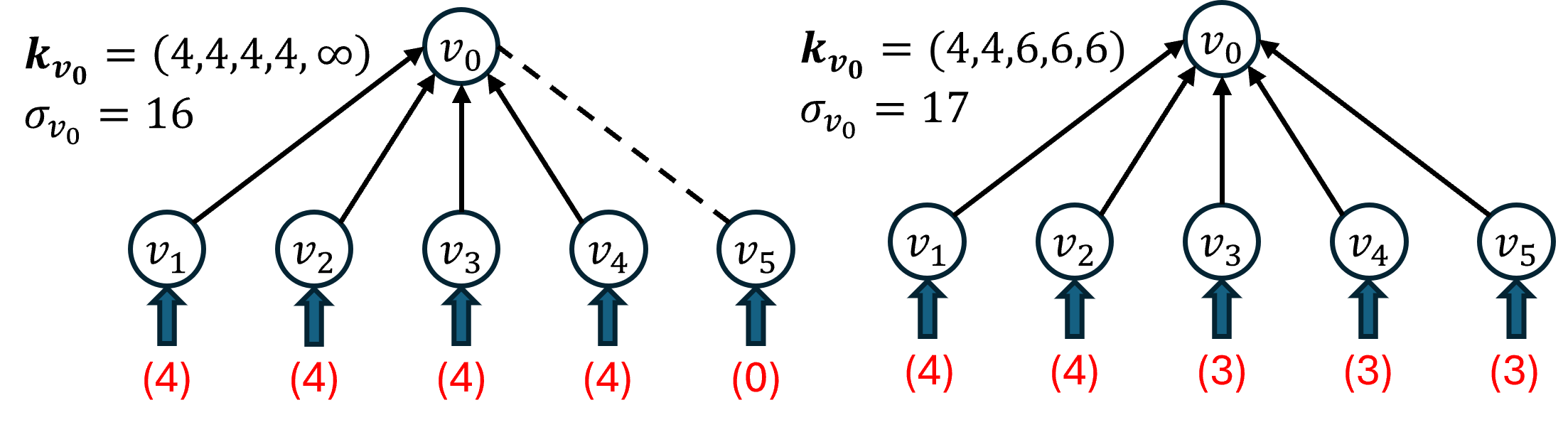}
    \caption{Best URR policy (left) vs. optimal policy (right)}
    \label{fig:example-sol}
\end{figure}

\section{Inductive Scheduling}\label{sec:inductivescheduling}

As discussed in the introduction, forming a schedule with constraints on maximum inter-scheduling times is known as pinwheel scheduling, and determining whether such a schedule exists for a given $\boldsymbol{k}$ is NP-complete~\cite{holte1989pinwheel}. Polynomial-time algorithms exist which can schedule certain classes of vectors, and the state of the art algorithm, known as $S_{xy}$ in the literature, is able to schedule all known polynomial-time schedulable vectors~\cite{chan_general_1992}. 

The schedulability of a vector $\boldsymbol{k} = (k_0,\dots,k_{M-1})$ is closely linked to a quantity known as its density and defined as 
\begin{equation}
    \rho(\boldsymbol{k}) \triangleq \sum_{i=0}^{M-1} 1/k_i.
\end{equation}
We further define the set of all vectors with densities less than or equal to $\rho$ as $\mathcal{K}_{\rho}$. It is easy to see that no vectors outside of $\mathcal{K}_1$ are schedulable, which provides an upper bound and a necessary condition on schedulable density. While this condition is far from sufficient, some vectors with density $1$ are schedulable. Consider a vector of $M$ elements, each equal to $M$. Vectors of this form have a density of $1$ and are trivially schedulable by a round-robin policy. More generally, however, sufficient conditions are difficult to obtain. It was recently shown that all vectors in $\mathcal{K}_{5/6}$ are schedulable~\cite{kawamura_proof_2024}, and this is the largest possible density bound. This can be seen by considering the vector $\boldsymbol{k} = (2,3,x)$, which has density $5/6 + 1/x$ and is not schedulable for any value of $x$.

While schedules are known to exist for all vectors in $\mathcal{K}_{5/6}$, there exists no known algorithm for scheduling all vectors in $\mathcal{K}_{5/6}$. The largest value of $\rho$ for which $\mathcal{K}_{\rho}$ is known to be polynomial-time schedulable is $0.7$, and all vectors in $\mathcal{K}_{0.7}$ can be scheduled by the $S_{xy}$ algorithm. This algorithm works by finding a vector $\boldsymbol{k'}$ where $k_i' \leq k_i$ for every $i$, as well as two integers $x$ and $y$ which need not be contained in $\boldsymbol{k'}$, but where each $k_i'$ takes the form $2^{\alpha_i} x$ or $2^{\alpha_i} y$ for some integer $\alpha_i$. Denote $\boldsymbol{k'_x}$ as the subvector whose elements are powers of two times $x$, and $\boldsymbol{k'_y}$ likewise with respect to $y$. Then if
\begin{equation}\label{eq:xyschedulability}
    \frac{\lceil x \rho(\boldsymbol{k'_x}) \rceil}{x} + \frac{\lceil y \rho(\boldsymbol{k'_y}) \rceil}{y} \leq 1,
\end{equation}
then $\boldsymbol{k'}$ is schedulable, and as a result $\boldsymbol{k}$ is schedulable. The $S_{xy}$ algorithm, which has complexity $O(M^2)$, searches for values of $\boldsymbol{k'}$, $x$, and $y$ which satisfy this condition, and if found produces a valid schedule. We encourage readers to reference~\cite{chan_general_1992} for complete details.

Although $S_{xy}$ represents the state of the art in pinwheel scheduling, it is easy to find schedulable vectors for which $S_{xy}$ does not return a schedule. One example is the vector $\boldsymbol{k} = (3,5,5,9,9)$, which has a density of $0.956$. Let $v_i$ represent the task corresponding to the inter-scheduling time $k_i$. One can verify that the cyclic schedule 
\begin{equation}\label{eq:examplepi}
    \pi = \{ v_0, v_1, v_2, v_0, v_3, v_1, v_0, v_2, v_4 \},
\end{equation}
which repeats itself every $9$ slots, satisfies the inter-scheduling times in $\boldsymbol{k}$, but $S_{xy}$ is unable to find this schedule.

To help close this gap, we introduce a novel pinwheel scheduling algorithm called \textit{Inductive Scheduling ($IS$)}. This algorithm uses $S_{xy}$ as a subroutine and is able to schedule all vectors schedulable by $S_{xy}$, while significantly expanding this set to include other vectors, including the example above. We first introduce the algorithm and then show how it can be used to solve $(P1)$.

The $IS$ algorithm is based on a technique of removing elements from a vector $\boldsymbol{k}$ one by one and augmenting the remaining elements such that if the augmented vector is schedulable then the original vector must also be schedulable. This process is continued until the augmented vector satisfies a known schedulability condition, or until the density of the augmented vector exceeds $1$. Once a known condition is satisfied, a schedule is generated for the augmented vector, and the elements which were removed are re-inserted into the schedule one by one. The resulting schedule is then guaranteed to satisfy the original inter-scheduling times.

We use superscripts to denote the value of $\boldsymbol{k}$ at each step, or iteration, of the algorithm. Let $\boldsymbol{k^0} = \boldsymbol{k}$ be the initial vector sorted in increasing order, and let $k_i^j$ be the $i$-th element at iteration $j$. To maintain a consistent vector length, we ``remove'' an element $m$ at iteration $j > 0$ by setting $k_m^j \to \infty$, where this is understood to mean that $v_m$ does not appear in the corresponding schedule $\pi^j$. Furthermore, note that $k_m^j = \infty$ has no effect on the density of $\boldsymbol{k^j}$.

Consider element $j$ at iteration $j$ and assume that a valid schedule $\pi^j$ exists for $\boldsymbol{k^j}$. Now further assume that $v_j$ is scheduled \textit{exactly} every $k_j^j$ slots in $\pi^j$. Note that the inter-scheduling time constraints only specify that $v_j$ be scheduled \textit{at least} every $k_j^j$ slots, so this is a more stringent condition. When inter-scheduling times are exact for $v_j$, we say that the schedule is regular with respect to $v_j$, and by imposing this constraint we say we are regularizing the schedule with respect to $v_j$.

\begin{lemma}\label{lemma:isinduction}
    Given a vector $\boldsymbol{k^j}$, let 
    \begin{equation}\label{eq:kevolution}
        k_i^{j+1} = \begin{cases}
            k_i^j - \big \lceil \frac{k_i^j}{k_j^j} \big \rceil, &i > j, \\
            \infty, &i = j.
        \end{cases}
    \end{equation}
    If $\boldsymbol{k^{j+1}}$ is schedulable, then $\boldsymbol{k^j}$ is also schedulable, and there exists a satisfying schedule $\pi^j$ which is regular with respect to $v_j$.
\end{lemma}

\begin{proof}
    If a vector $\boldsymbol{k^{j+1}}$ is schedulable, then each element $i$ must appear at least once every $k_i^{j+1}$ slots in the schedule $\pi^{j+1}$, so
    \begin{equation}
        \pi_{i,l}^{j+1} - \pi_{i,l-1}^{j+1} \leq k_i^{j+1}, \ \forall \ 0 \leq i < M, l \geq 1.
    \end{equation}
    Suppose we reintroduce $v_m$ back into the schedule by inserting it at slots $0, k_m^j, 2k_m^j, \dots$, and pushing subsequent elements back a slot each time it is inserted. The resulting schedule $\pi^j$ is by definition regular with respect to $v_m$.

    Now consider the remaining elements. Because the initial vector was sorted, each tasks $v_i$ where $i < m$ with finite $k_i^{j+1}$ has $k_i^{j+1} < k_m^j$. Therefore, at most one occurrence of $v_m$ can occur between consecutive occurrences of $v_i$, and
    \begin{equation}
        \pi_{i,l}^j - \pi_{i,l-1}^j \leq k_i^{j+1}+1 = k_i^j, \ \forall \ l \geq 1.
    \end{equation}

    Similarly, each task $v_i$ with $i \geq m$ and finite $k_i^{j+1}$ has at most $n_m^j = \lceil (k_i^{j+1} + n_m^j) / k_m^j \rceil$ occurrences of $v_m$ between consecutive occurrences of $v_i$. Then 
    \begin{equation}
        \pi_{i,l}^j - \pi_{i,l-1}^j \leq k_i^{j+1} + n_m^j = k_i^{j+1} + \Big \lceil \frac{k_i^{j+1} + n_m^j}{k_m^j} \Big \rceil, \ \forall \ l \geq 1.
    \end{equation}
    Rearranging yields 
    \begin{equation}
        k_i^{j+1} = k_i^{j+1} + n_m^j - \Big \lceil \frac{k_i^{j+1} + n_m^j}{k_m^j} \Big \rceil.
    \end{equation}
    From the definition of $k_i^{j+1}$ in~\eqref{eq:kevolution}, this equation has a solution with $k_i^{j+1} + n_m^j = k_i^j$.

    Therefore, $\pi_{i,l}^j - \pi_{i,l-1}^j \leq k_i^j$ for all $i$ with finite $k_i^j$ and all $l \geq 1$, and $\pi^j$ satisfies the inter-scheduling conditions of $\boldsymbol{k^j}$. This completes the proof.
\end{proof}

This lemma provides the crucial step in the $IS$ algorithm. If a vector $\boldsymbol{k}$ is not schedulable by any known methods, but after regularizing with respect to $v_j$ the resulting vector \textit{is} schedulable, then Lemma~\ref{lemma:isinduction} shows that we can schedule the original vector. Consider the example from above with $\boldsymbol{k} = (3,5,5,9,9)$ and recall that this is not schedulable by $S_{xy}$ or any known algorithm in the literature. If we regularize with respect to $v_0$, this yields $\boldsymbol{k^1} = (\infty,3,3,6,6)$. This has density $1$ and consists of only two distinct finite values, which is sufficient to be schedulable by $S_{xy}$ with 
\begin{equation}
    \pi^1 = \{ v_1, v_2, v_3, v_1, v_2, v_4 \}.
\end{equation}
Now re-inserting $v_0$ at slots $0$, $k_0 = 3$, and $2k_0 = 6$ yields 
\begin{equation}
    \pi = \{ v_0, v_1, v_2, v_0, v_3, v_1, v_0, v_2, v_4 \},
\end{equation}
which satisfies the original vector and matches~\eqref{eq:examplepi}.

In general, we may have to carry out several iterations of $IS$ before arriving at a schedulable vector. The vector $\boldsymbol{k}=(3,5,8,8,14,14)$ for example is not schedulable by $S_{xy}$, but regularizing with respect to $v_0$ yields $\boldsymbol{k^1} = (\infty,3,5,5,9,9)$. From the example above we know this is schedulable after regularizing again with respect to $v_1$. In the worst case, $IS$ must run for $M-2$ iterations. Testing schedulability at each iteration has complexity $O(M^2)$, yielding a total complexity of $O(M^3)$.

\begin{theorem}\label{th:isalgorithm}
    Denote the set of vectors schedulable by the Inductive Scheduling algorithm as $\mathcal{K}_{IS}$. Then $\mathcal{K}_{IS} \supset \mathcal{K}_{\mathcal{A}}$ for any known polynomial-time algorithm $\mathcal{A}$.
\end{theorem}

\begin{proof}    
    This follows from the fact that all previously known polynomial-time schedulable vectors are schedulable by $S_{xy}$,. Because $IS$ tests for $S_{xy}$ schedulability at each iteration, it must schedule at least all of $\mathcal{K}_{S_{xy}}$. It is a strict superset by the existence of the vector $\boldsymbol{k}=(3,5,8,8,8)$, which is schedulable by $IS$ and is not schedulable by $S_{xy}$. Therefore, there exist vectors in $\mathcal{K}_{IS}$ which are not in $\mathcal{K}_{\mathcal{A}}$, for any known algorithm $\mathcal{A}$.
\end{proof}

While it is difficult to characterize $\mathcal{K}_{IS}$ or provide general density guarantees beyond the $0.7$ guarantee given by $S_{xy}$, numerical results in Section~\ref{sec:simulations} show that $IS$ significantly outperforms $S_{xy}$ on over $1.4$ million randomly generated vectors with densities between $0.7$ and $1$. Most notably, the smallest density of any vector found to be outside of $\mathcal{K}_{IS}$ was $0.843$, compared to $0.774$ with $S_{xy}$. Recall that any vector with density up to $5/6$ is known to be schedulable, though not in polynomial time. \textit{Every vector} which met this bound was scheduled by $IS$, which leads us to the following conjecture. We cannot claim with certainty that the conjecture holds, but we have not found a counterexample.

\begin{conjecture}
    The $IS$ algorithm can schedule all vectors with densities up to $5/6$.
\end{conjecture}

In the next section, we will show how $IS$ can be used to solve the utility maximization problem $(P1)$. Before doing so, it is instructive to consider the role density plays in the optimization. Recall that the solution to $(P1)$ is truly optimal only as far as we can optimize over all feasible schedules. By scheduling a larger class of high-density vectors, $IS$ shrinks this optimality gap.

Assume that capacities are symmetric across each level of the tree, and that we have found a solution with schedulable $\boldsymbol{k_v}$ at each node $v$ that satisfies deadline constraints. All flows which are supported must pass through some node at level $d$, so $\sigma^* = \sum_{v \in V_d} \sigma_v^*$, for any $d$. Combining this with slice capacity and flow conservation constraints,
\begin{align}
\begin{aligned}
    \sigma^* &= \sum_{v \in V_d} \sigma_v^* = \sum_{v \in V_d} \sum_{v' \in \mathcal{C}_v} \sigma_{v'}^* \\
    &\leq \min_{0 \leq d < D} \sum_{v \in V_d} \sum_{v' \in \mathcal{C}_v} \Big \lfloor \frac{c_{e(v')}}{\lambda} \Big ( \frac{1}{k_{e(v')}} \Big) \Big \rfloor \\
    &\leq \min_{0 \leq d < D} \sum_{v \in V_d} \frac{c_{d+1}}{\lambda} \rho(\boldsymbol{k_v}),
\end{aligned}
\end{align}

This bound is not necessarily tight, but by scheduling vectors with larger densities, $IS$ increases this bound and helps ensure that utility is limited by the network itself and not by scheduling limitations. More concretely, there are at least some instances where the network is saturated up to level $d$ and this bound is tight (up to the floor). In these instances, $IS$ increases utility proportionally to its increase in scheduled density over $S_{xy}$. 

\section{General Utility Maximization}\label{sec:generalarrivals}

We finally turn to solving our utility maximization problem $(P1)$ by applying schedulability conditions from the $IS$ algorithm. This is achieved by introducing an additional variable into the optimization for each variable used in the algorithm, along with constraints describing the variables' relationships. While the $IS$ algorithm appears relatively straightforward, there are a number of challenges in this process that we will address in turn. Our goal is to express each constraint in a linear fashion, so that we are left with a tractable mixed-integer linear program (MILP). The derivation is quite tedious, so it can be skipped if the reader is not interested in the details. We present the full derivation in Appendix~\ref{app:milpderivation}.

While quite cumbersome, the derived variables and linear constraints, combined with the linear constraints already present in $(P1)$, fully define a MILP that maximizes $\sigma$ over all admission decisions and feasible $IS$ schedules. At each node $v$ in the formulation, there are $O(N_v^2)$ binary variables, one for each child being scheduled at each possible $IS$ iteration. Each variable has a fixed number of constraints, so there are likewise $O(N_v^2)$ linear constraints. In total, this gives us a MILP representation of $(P1)$ with $O\big( \sum_{v \in V_c} N_v^2 \big)$ binary variables and $O\big( \sum_{v \in V_c} N_v^2 \big)$ linear constraints, where $V_c$ is the set of vertices in the tree with children, i.e., all vertices except the leaves.

Modern mixed-integer solvers like Gurobi have become quite powerful at solving well-structured MILPs with several hundred binary variables, with typical solve times ranging from less than a second to a few minutes on a standard laptop. Worst-case complexity, however, is always exponential in the number of binary variables, and the solve times can also grow exponentially. As a result, solving this MILP directly is only practical on small problem instances, where $V_c$ cannot grow too large. 

We develop a scalable solution to the problem by decomposing the global optimization into a series of subproblems at each node $v$, where each subproblem retains its fixed $O(N_v^2)$ binary variables and $O(N_v^2)$ linear constraints. These subproblems are tractable for reasonable values of $N_v$, and the total solve time grows linearly, rather than exponentially, in the size of $V_c$. We explore this decomposition next.

\subsection{Distributed SLA Utility Maximization}

The bulk of the problem complexity lies in optimizing over the $IS$ schedulability constraints, which are independent between nodes under our local interference model. A close examination of $(P1)$ shows that the $\boldsymbol{k_v}$ vectors are coupled only through the deadline constraints and the flow conservation constraints. We will show that these can also be decoupled by solving for the maximum flow at each node independently.

Assume that deadline constraints are relaxed and consider flow conservation. The number of flows $\sigma_v^*$ which pass through node $v$ under a utility-maximizing policy is a function of the scheduling decisions along each leaf-to-root path that contains $v$, so this quantity is coupled between nodes. Now define the maximum number of flows supported by the subtree rooted at node $v$, independent of the rest of the tree, as
\begin{equation}
    \hat{\sigma}_v^* = \max_{\boldsymbol{k_v}} \sum_{v' \in \mathcal{C}_v} \sigma_{v'},
\end{equation}
subject to $\sigma_{v'} \leq \hat{\sigma}_{v'}^*$, for all $v' \in \mathcal{C}_v$, and slice capacity and $IS$ constraints. Note the distinction between this quantity and $\sigma_v^*$. The former represents an upper bound on the latter, and is independent of scheduling decisions above it in the tree.

Now consider the deadline constraints $\sum_{e \in T} k_e \leq \tau, \ \forall \ T \in \mathcal{T}$. At first glance, these constraints appear to pose a serious challenge to our decomposition because they directly couple the inter-scheduling times along each leaf-to-root path. The key observation is that only the sum $\sum_{v' \in V_v^+} k_{e(v')}$ affects $\hat{\sigma}_v^*$, where recall that $V_v^+$ denotes the set of node $v$'s ancestors, and the individual values of each $k_{e(v')}$ are irrelevant.

To see this, define the quantity
\begin{equation}\label{eq:tauvdef}
    \tau_v \triangleq \tau - k_{e(v)} - \sum_{v' \in V_v^+} k_{e(v')},
\end{equation}
which represents the effective deadline by which packets must reach node $v$. When $\tau_v$ is fixed, the deadline constraints can be rewritten as
\begin{equation}
    \sum_{v' \in T \cap V_v^-} k_{e(v')} \leq \tau_v, \ \forall \ T \in \mathcal{T},
\end{equation}
which shows that $\boldsymbol{k_v}$ only depends on the quantity $\tau_v$ and vectors $\boldsymbol{k_{v'}}$ below it in the tree.

Therefore, when $\tau_v$ is fixed, $\hat{\sigma}_v^*$ can be computed in the following way, using only knowledge of the scheduling decisions below node $v$ in the tree. We use the notation $\hat{\sigma}_v^*(\tau_v)$ to show the dependence on $\tau_v$.

\begin{lemma}\label{lemma:P2}
    Let $\tau_v$ and $\hat{\sigma}_{v'}^*(\tau_{v'})$ be fixed, for all $v' \in \mathcal{C}_v$ and all feasible $\tau_{v'}$. Then $\hat{\sigma}_v^*(\tau_v)$ is the solution to 
    \begin{align}
    \begin{aligned}\label{eq:distributedopt}
        (P2) : \max_{\boldsymbol{k_v}, \pi_v} \ &\sum_{v' \in \mathcal{C}_v} \sigma_{v'} \\ 
        \text{s.t.} \ &\sigma_{v'} \leq \hat{\sigma}_{v'}^*(\tau_v - k_{e(v')}), \ \forall \ v' \in \mathcal{C}_v, \\ 
        &\sigma_{v'} \lambda k_{e(v')} \leq c_{e(v')}, \ \forall \ v' \in \mathcal{C}_v, \\ 
        &\pi_v = IS(\boldsymbol{k_v}), \\ 
        &\sigma_{v'}, k_{e(v')} \in \mathbb{Z}_+, \ \forall \ v' \in \mathcal{C}_v.
    \end{aligned}
    \end{align}
\end{lemma}

\begin{proof}
    From~\eqref{eq:tauvdef},
    \begin{equation}
        \tau_{v'} = \tau - k_{e(v')} - \sum_{v'' \in V_{v'}^+} k_{e(v'')} = \tau_v - k_{e(v')},
    \end{equation} 
    so for a fixed $\tau_v$ and scheduling vector $\boldsymbol{k_v}$, the maximum number of flows supported at $v'$ is $\hat{\sigma}_{v'}^*(\tau_v - k_{e(v')})$. This recursively handles the deadline constraints through the value of $\tau_v - k_{e(v')}$. The remaining constraints handle slice capacity and schedulability, so maximizing over $\boldsymbol{k_v}$ yields the optimal solution for a fixed $\tau_v$.
\end{proof}

Note the differences between $(P1)$ and $(P2)$. While $(P1)$ is a global optimization that solves for $\sigma_v^*$ simultaneously at each $v$, $(P2)$ is a local optimization, which solves for a particular $\hat{\sigma}_v^*$. The slice capacity and $IS$ constraints are the same because these constraints are decoupled between nodes, while the deadline constraints from $(P1)$ are implicitly present in the $(P2)$ flow conservation constraints $\sigma_{v'} \leq \hat{\sigma}_{v'}^*(\tau_v - k_{e(v')}), \ \forall \ v' \in \mathcal{C}_v$. To see this, observe that when $\tau_v$ and $k_{e(v')}$ are fixed, $\tau_{v'} = \tau_v - k_{e(v')}$, and the constraint becomes $\sigma_{v'} \leq \hat{\sigma}_{v'}^*(\tau_{v'})$, which enforces both flow conservation and deadlines.

Solving $(P2)$ at node $v$ requires knowledge of $\hat{\sigma}_{v'}^*$ at its children, which naturally lends itself to an algorithm which begins at the APs and iterates up the tree until reaching the root. However, because $\tau_v$ depends on the scheduling decisions at nodes \textit{above} $v$ in the tree, the algorithm cannot know the optimal value $\tau_v^*$ a priori. Instead, it solves $(P2)$ once for each feasible value of $\tau_v$, and returns the set
\begin{equation}
    \{ \hat{\sigma}_v^*(\tau_v), \ \forall \ D - d(v) \leq \tau_v \leq \tau - d(v) \}
\end{equation}
to its parent, which forms the upper bound in the flow conservation constraint at its parent node. This encompasses all feasible values of $\tau_v$ because $k_{e(v')} \geq 1$ at each $v' \in V$, so the algorithm must solve $(P2)$ a total of $\tau-D$ times at each node.

Once the algorithm has solved $(P2)$ at the root, it propagates the optimal solution back down the tree by iteratively finding $\hat{\sigma}_v^*(\tau_v^*)$ from the pre-computed values at node $v$, and then setting $\tau_{v'}^* = \tau_v^* - k_{e(v')}^*$ for each $v' \in \mathcal{C}_v$. 

While $(P2)$ contains the same non-linear and non-convex constraints as $(P1)$, they can be linearized in exactly the same way, resulting in a MILP with $O(N_v^2)$ binary variables and $O(N_v^2)$ linear constraints. The algorithm then solves at most $(\tau-D) V_c$ iterations of this linearized version to reach the global optimum. We can reduce the dependency on $V_c$ even further, however, by solving in a distributed fashion. Because each subproblem only depends on values below it in the tree, we can parallelize the computations across each level, resulting in a total computation time that is linear in the tree depth $D$ rather than $V_c$. We describe the full Distributed SLA Utility Maximization ($DSUM$) algorithm in Algorithm~\ref{alg:distributedopt}.

\begin{algorithm}
    \DontPrintSemicolon
    %
    \SetKwInput{Input}{Input}\SetKwInOut{Output}{Output}
    \Input{Flow arrivals $N_{D,v}, \ \forall \ v \in V_{D-1}$, rate $\lambda$, deadline $\tau$}
    \Output{Maximum utility $\hat{\sigma}^*$, optimal policy $\pi^*$}
    \For{$v \in V_{D-1}$}{
        \For{$\tau_v \in [1,\tau-D+1]$}{
            Set $\hat{\sigma}_v^*(\tau_v) = \min \big \{ N_{D,v}, \tau_v, \big \lfloor \frac{c_D}{\lambda} \big \rfloor \big \}$
        }
    }
    $d = D-2$ \\
    \While{$d > 0$}{
        \For{parallel $v \in V_d$}{
            \For{$\tau_v \in [1,\tau-d]$}{
                Set $\hat{\sigma}_v^*(\tau_v), \boldsymbol{k_v^*}(\tau_v)$ from $(P2)$
            }
        }
        Decrement $d$ \\
    }
    Set $\hat{\sigma}_{root}^*(\tau)$ from $(P2)$ \\
    $\tau_{root}^* = \tau$ \\
    \While{$d < D$}{
        \For{parallel $v \in V_d$}{
            \For{$v' \in \mathcal{C}_v$}{
                $k_{e(v')}^* = k_{e(v')}^*(\tau_v^*)$ \\
                $\tau_{v'}^* = \tau_v^* - k_{e(v')}^*$ \\
            }
            $\pi_v = IS(\boldsymbol{k_v})$ \\
        }
        Increment $d$ \\
    }
    Return $(\hat{\sigma}_{root}^*, \pi^* = \{ \pi_v, \ \forall \ v \in V_c \})$ \\   
    
    \caption{Distributed SLA Utility Maximization (DSUM)}
    \label{alg:distributedopt}
\end{algorithm}

\begin{theorem}\label{th:dsum}
    $DSUM$ returns the solution to $(P1)$, maximizing utility over all admission decisions and $IS$ scheduling policies. It does this by setting 
    \begin{equation}
        \hat{\sigma}_v^*(\tau_v) = \min \big \{ N_{D,v}, \tau_v, \big \lfloor \frac{c_D}{\lambda} \big \rfloor \big \}
    \end{equation}
    at each AP $v$, and for all feasible $\tau_v$, and then solving at most $\tau (D-1)$ sequential batches of $(P2)$, where each batch consists of parallel $(P2)$ computations across nodes on one level of the tree.
\end{theorem}

\begin{proof}
    From Lemma~\ref{lemma:P2}, if $\hat{\sigma}_{v'}^*(\tau_{v'})$ is known for all $v' \in \mathcal{C}_v$ and feasible $\tau_{v'}$, then $(P2)$ returns $\hat{\sigma}_v^*(\tau_v)$ for a fixed $\tau_v$. Therefore, by induction, we can solve $\hat{\sigma}_{root}^*(\tau)$ given an initial solution at the APs. 

    At each AP $v$, there are $N_{D,v}$ customers, each requesting an SLA for a single flow. Therefore, $\hat{\sigma}_{v'}^*(\tau_v - k_{e(v')}) = 1$ if $\tau_v \geq k_{e(v')}$ and $0$ otherwise, for each customer $v'$. Then $(P2)$ simplifies to 
    \begin{align}
    \begin{aligned}\label{eq:distributedopt}
        \max_{\boldsymbol{k_v}} \ &\sum_{v' = 0}^{N_{D,v} - 1} \sigma_{v'} \\ 
        \text{s.t.} \ &\sigma_{v'} = \mathbbm{1}(\tau_v \geq k_{e(v')}), \ \forall \ v' \in \mathcal{C}_v, \\ 
        &\sigma_{v'} \lambda k_{e(v')} \leq c_{e(v')}, \ \forall \ v' \in \mathcal{C}_v, \\ 
        &\boldsymbol{k_v} \ \text{schedulable by} \ IS, \\
        &k_{e(v')} \in \mathbb{Z}_+, \ \forall \ v' \in \mathcal{C}_v.
    \end{aligned}
    \end{align}
    Combining the first and second constraints yields
    \begin{equation}
        \sigma_{v'} = \mathbbm{1} \Big( k_{e(v')} \leq \min \Big \{ \tau_v, \frac{c_{e(v')}}{\lambda} \Big \} \Big ),
    \end{equation}
    and because $c_{e(v')} = c_D$ across all $v'$, this constraint is symmetric. $IS$ can always schedule a vector with $k_{e(v')}$ elements equal to $k_{e(v')}$ using a round-robin schedule, so setting 
    \begin{equation}
        k_{e(v')} = \min \Big \{ \tau_v, \Big \lfloor \frac{c_D}{\lambda} \Big \rfloor \Big \}
    \end{equation}
    yields $\hat{\sigma}_v^* = \min \{ \tau_v, \lfloor c_D / \lambda \rfloor \}$, provided there are enough customers at the AP. If $N_{D,v}$ is less than this value, then the same policy can support $N_{D,v}$ flows, so the solution at any AP is given by 
    \begin{equation}
        \hat{\sigma}_v^*(\tau_v) = \min \Big \{ N_{D,v}, \tau_v, \Big \lfloor \frac{c_D}{\lambda} \Big \rfloor \Big \}.
    \end{equation}

    The optimal solution can be found in $\tau (D-1)$ sequential batches because each node on the same level can solve $(P2)$ in parallel and each node must solve a maximum of $\tau$ iterations. The solution at the APs exists in closed form so we need only solve at the $D-1$ levels above it. This completes the proof.
\end{proof}

The $DSUM$ algorithm reduces the computation from solving a MILP with $O(\sum_{v \in V_c} N_v^2)$ binary variables to solving at most $\tau (D-1)$ batches of parallel MILPs with $O(N_v^2)$ binary variables each. In the next section, we will show numerical results including computation times, but in general $DSUM$ can be solved on the order of minutes for reasonable values of $N_v$. As such, $DSUM$ offers a practical and provably near-optimal solution to maximizing utility with long-lasting SLAs. 

\section{Numerical Results}\label{sec:simulations}

\subsection{IS Algorithm}

In this section we present numerical results highlighting the performance of the $IS$ and $DSUM$ algorithms. To quantify the improvement of the $IS$ algorithm over $S_{xy}$, we randomly generated a large number of vectors and tested the performance of both algorithms on each one. These vectors varied in length from $4$ to $20$ elements. For each length $M$, we generated $100,000$ unique vectors with densities between $0.7$ and $1$ using the following method\footnote{For vector lengths between $4$ and $7$, we were unable to find this many unique vectors, so the search was truncated for each length after $100,000$ consecutive failed attempts.}. First, each vector element was generated from a uniform distribution in the range $[2,3M-1]$. If the density of this vector was outside of the range $(0.7,1]$, or if the sorted vector had been seen before, then it was discarded and the process was repeated until a new vector with proper density was found. Next, this vector was given to both the $S_{xy}$ and $IS$ algorithms to determine if it was schedulable.

The average success rate of both algorithms is shown on the left of Figure~\ref{fig:is-sxy-comp}. For $M=4$, $IS$ did not present any improvement over $S_{xy}$, but starting at length $5$ the performance gap widened, and for $M \geq 8$, $IS$ consistently scheduled from $19\%$ to $24\%$ more vectors than $S_{xy}$.

On the right of Figure~\ref{fig:is-sxy-comp} we show the minimum density which was not schedulable for each vector length, and see that for $M \geq 7$, the minimum $IS$ density was consistently $0.05$ to $0.07$ higher than $S_{xy}$. In addition, out of over $1.4$ million vectors that were tested, the smallest density which was not schedulable by $IS$ was $0.834$. This supports our conjecture that $IS$ can schedule all vectors with density up to $5/6$ (approximatley $0.833$), which recall is the largest possible density bound. To further support this conjecture, we repeated the experiment, this time constricting density to be in the range $(0.7,0.83]$, and we generated over $1.3$ million unique vectors in this experiment. As expected, $IS$ was able to schedule every one of these vectors.

\begin{figure}
    \centering
    \includegraphics[width=0.48\textwidth]{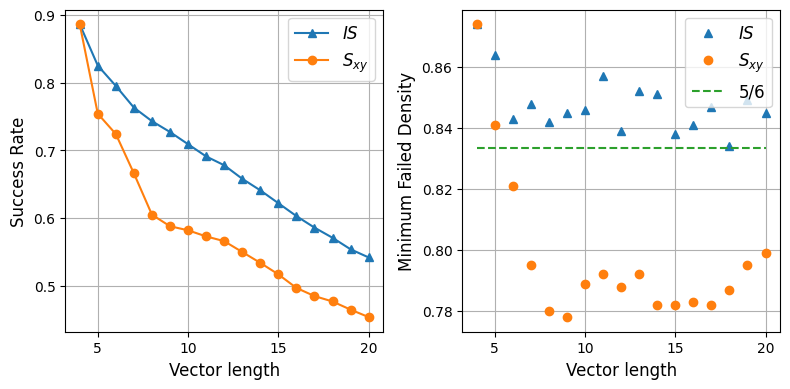}
    \caption{Success rate and minimum unschedulable density of the $IS$ and $S_{xy}$ algorithms}
    \label{fig:is-sxy-comp}
\end{figure}

\subsection{DSUM Algorithm}

We tested the performance of the $DSUM$ algorithm on various trees of depth $3$, while varying the link capacities, the number of children (i.e., the degree) at each node, and the number of flows at each AP. For the purposes of these simulations, we set $\lambda = 10$ Mbps for each flow, and the duration of a time slot as $125 \ \mu s$, the smallest time slot used for data in the 5G standard~\cite{3gpp-ts-138211}. Note that both $\lambda$ and the time slot duration can be scaled without changing the results, provided the link capacities are scaled with $\lambda$.

\begin{figure}
    \centering
    \includegraphics[width=0.48\textwidth]{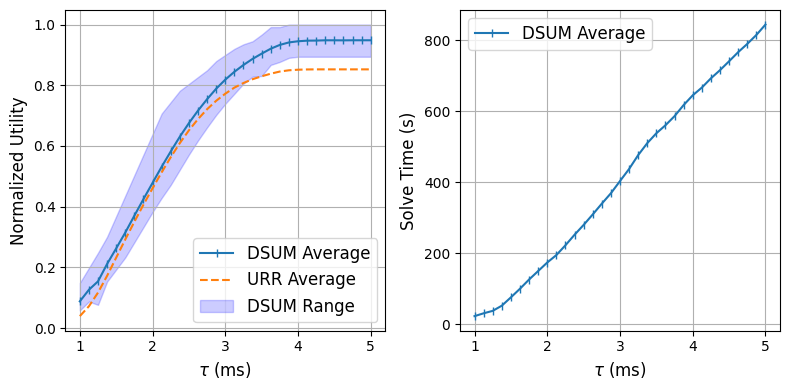}
    \caption{Normalized utility and solve times for the $DSUM$ algorithm on trees with depth $3$ and varying degrees from $2$ to $6$, with $400$ flows requesting SLAs}
    \label{fig:dsum-perf}
\end{figure}

For each experiment, we constructed a depth $3$ tree with a random number of children at each node, uniformly distributed on the interval $[2,6]$. The uplink capacity from each user to its AP was fixed at $250$ Mbps, and the maximum capacities of the level-$2$ and level-$1$ links were fixed at $1$ Gbps and $4$ Gbps respectively. When each tree was constructed, the actual capacity of each level-$2$ and level-$1$ link was drawn uniformly from the range $[C_{max} /2, C_{max}]$, where $C_{max}$ represents the maximum capacity for the link as described above. We fixed the number of flows requesting SLAs at $400$, and randomly distributed them at the APs following a uniform distribution.

Once each tree was constructed, we ran the $DSUM$ algorithm on the tree and observed the utility the algorithm was able to achieve, as well as the solve time of the algorithm. We also tested the performance of the $URR$ algorithm with greedy pruning. The results of $50$ such experiments are shown in Figure~\ref{fig:dsum-perf} as a function of $\tau$. The left side of the figure shows the utility averaged over each experiment, as well as the maximum and minimum values, with all utilities normalized by the tree capacity. Note that the capacity is an upper bound on utility, and may not be achievable for bounded deadlines and schedule length. Nevertheless, we observe that the average $DSUM$ performance is approximately $0.95$ of this bound as $\tau$ becomes large, and achieves the bound with equality in some instances.

In general, we observe an inflection point at a certain value of $\tau$, which splits the curve into two regions. For small values of $\tau$, deadline constraints dominate, and utility increases roughly linearly with $\tau$. Once a certain value of $\tau$ is reached, capacity constraints begin to dominate, and the curve flattens. This also explains the performance of the $URR$ policy, plotted alongside $DSUM$. Recall that $URR$ is optimal under perfect symmetry conditions. When $\tau$ is small, we must prune a significant number of branches from the tree in order to satisfy deadlines, and as discussed in the context of the Greedy Pruning algorithm, pruning at a node with larger degree has a smaller average impact on utility. Therefore, pruning the tree optimally tends to drive the tree to a more symmetric state, where $URR$ is close to optimal. As $\tau$ reaches the capacity-dominated region, $DSUM$ exploits this by deviating from $URR$ and achieving approximately $10\%$ higher utility.

The right side of Figure~\ref{fig:dsum-perf} shows the average solve times for $DSUM$ as a function of $\tau$, averaged over the same set of experiments. These experiments were run on a standard laptop with an Intel Core i7-1365U processor. Each subproblem was solved sequentially and the maximum solve time at each level of the tree was summed to compute the total distributed solve time, which is shown in the figure. The solve times increase roughly linearly with $\tau$ as expected, because each node solves approximately $\tau$ versions of $(P2)$. We observe that solve times for the complete algorithm are on the order of minutes for depth $3$ trees with degree up to $6$ and $400$ SLAs.

\section{Conclusion}\label{sec:conclusion}

In this paper, we studied the problem of maximizing utility from SLAs, defined in terms of instantaneous throughput and hard packet deadlines. We showed that under perfect symmetry conditions, a URR policy is optimal for maximizing utility, but that in the general case a more complete framework is necessary. We introduced a novel pinwheel scheduling algorithm, which significantly outperforms the state of the art. Using conditions from this algorithm, we formulated a MILP to solve the utility maximization problem, and developed a scalable distributed algorithm to solve this MILP, with solve times that grow linearly in the depth of the tree. Future extensions of this work involve extending the $IS$ and $DSUM$ algorithms to support multiple simultaneous transmissions, which is applicable for technologies like OFDM with multiple subcarriers, and supporting non-batch flow arrivals.

\appendix

\subsection{Derivation of the MILP (P1) equivalent}\label{app:milpderivation}

For compactness, we will use vector and matrix notation in this section, denoting vectors in bold and matrices with capital letters. In addition, we make use of the notation $\mathbf{1_M}$ where possible to denote a vector of ones with length $M$. Before introducing the $IS$ constraints, we first address the nononvex slice capacity constraint~\eqref{subeq:slicecap}, which has a bilinear quadratic term. To linearize this constraint, we employ the well-known McCormick envelope technique~\cite{mccormick1976computability}, which introduces additional variables and constraints to linearize bilinear terms when at least one of the variables is binary. We will employ this technique several times going forward. Let $\boldsymbol{\kappa_e}$ be the binary representation of $k_e$, such that
\begin{equation}\label{eq:binaryk}
    k_e = \sum_{l=0}^{\lfloor \log k_{\max} \rfloor} 2^l \kappa_{el}.
\end{equation}
Then the constraint can be rewritten as 
\begin{equation}\label{eq:slicecapbin}
    \lambda \sigma_v \sum_{l=0}^{\lfloor \log k_{\max} \rfloor} 2^l \kappa_{e(v),l} \leq c_{e(v)},
\end{equation}
where the left hand side is the sum of products of an integer with a binary variable.

We can now apply the envelope technique, which uses upper and lower bounds on the non-binary term $\lambda \sigma_v$. Clearly a lower bound is $0$ and an upper bound is $c_{e(v)}$, because $k_{e(v)}$ cannot be less than $1$. Define $\gamma_{vl} \triangleq \lambda \sigma_v \kappa_{e(v),l}$, which is equivalent to the set of linear constraints
\begin{align}
\begin{aligned}
    &\gamma_{vl} \geq 0, \\
    &\gamma_{vl} \leq \lambda \sigma_v, \\
    &\gamma_{vl} \leq \kappa_{e(v),l} c_{e(v)}, \\
    &\gamma_{vl} \geq \lambda \sigma_v - c_{e(v)} (1 - \kappa_{e(v),l}).
\end{aligned}
\end{align}
Rewriting~\eqref{eq:slicecapbin} in terms of $\gamma_{vl}$, we have 
\begin{equation}
    \sum_{l=0}^{\lfloor \log k_{\max} \rfloor} 2^l \gamma_{vl} \leq c_{e(v)},
\end{equation}
which is now linear and equivalent to~\eqref{subeq:slicecap}.

Moving on to the $IS$ algorithm, the first difficulty we must address is that $IS$ requires $\boldsymbol{k}$ to be sorted, and we do not know this order a priori before solving $(P1)$. Therefore, whatever order we assign is not necessarily the order the algorithm requires. Recall that we denote the initial sorted vector as $\boldsymbol{k^0}$, so that $\boldsymbol{k^0} = \Phi \boldsymbol{k}$, where $\Phi$ is a permutation matrix (i.e., each row and each column contains exactly one $1$ and the remaining elements are $0$). In other words, if $\Phi_{ij}=1$, then $k_i^0 = k_j$, and the matrix represents a mapping between the vectors. Expanding the expression above,
\begin{equation}\label{eq:kpermutation}
    k_i^0 = \sum_{j=0}^{M-1} \Phi_{ij} k_j, \ \forall \ 0 \leq i < M,
\end{equation}
which is again the sum of bilinear terms with one binary factor $\Phi_{ij}$.

Similar to above, define $\zeta_{ij} \triangleq \Phi_{ij} k_j$, and replace this equality with the equivalent linear inequalities
\begin{align}
\begin{aligned}
    &\zeta_{ij} \geq 0, \\
    &\zeta_{ij} \leq k_j, \\
    &\zeta_{ij} \leq \Phi_{ij} k_{\max}, \\
    &\zeta_{ij} \geq k_j - k_{\max} (1 - \Phi_{ij}).
\end{aligned}
\end{align}
Replacing the bilinear term in~\eqref{eq:kpermutation} with $\zeta_{ij}$, and combining with 
\begin{align}
\begin{aligned}
    &\Phi \mathbf{1_M} = \mathbf{1_M}, \\
    &\Phi^T \mathbf{1_M} = \mathbf{1_M}, \\
    &k_i^0 \leq k_{i+1}^0, \ \forall \ 0 \leq i < M-1
\end{aligned}
\end{align}
gives us a linear mapping from $\boldsymbol{k}$ to the sorted vector $\boldsymbol{k^0}$.

Next we consider the evolution equation~\eqref{eq:kevolution}, and note that it is again both nonlinear and nonconvex. To linearize this expression, we define the variable $r_{ij} \triangleq \lceil k_i^j / k_j^j \rceil$, which is equivalent to 
\begin{equation}\label{eq:kevolratio}
    r_{ij}^j k_j^j \geq k_i^j, \ r_{ij} \in \mathbb{Z}.
\end{equation}
This formulation eliminates the ceiling operation, but is still nonconvex due to the bilinear quadratic term. Analogous to the binary representation of $k_e$ above, let $\boldsymbol{\kappa^j_i}$ be the binary representation of $k_i^j$, such that
\begin{equation}
    k_i^j = \sum_{l=0}^{\lfloor \log k_{\max} \rfloor} 2^l \kappa_{il}^j.
\end{equation}
Then~\eqref{eq:kevolratio} can be rewritten as 
\begin{equation}\label{eq:rkappaexp}
    r_{ij}^j \sum_{l=0}^{\lfloor \log k_{\max} \rfloor} 2^l \kappa_{jl}^j \geq k_i^j,
\end{equation}
where the left hand side is again the sum of products of an integer with a binary variable.

Again applyling the envelope technique, define $\xi_{ijl}^j \triangleq r_{ij}^j \kappa_{jl}^j$ and recall that $r_{ij}^j$ represents the ratio of two positive values of $k_i^j$, so it is upper bounded by $k_{\max}$ and lower bounded by zero. Then we have the familiar set of equivalent linear constraints
\begin{align}
\begin{aligned}
    &\xi_{ijl}^j \geq 0, \\
    &\xi_{ijl}^j \leq r_{ij}^j, \\
    &\xi_{ijl}^j \leq \kappa_{jl}^j k_{\max}, \\
    &\xi_{ijl}^j \geq r_{ij}^j - k_{\max} (1 - \kappa_{jl}^j).
\end{aligned}
\end{align}
Replacing the bilinear term in~\eqref{eq:rkappaexp} with $\xi_{ijl}^j$ yields a set of linear constraints equivalent to~\eqref{eq:kevolratio} for all $i > j$ at iteration $j$. The evolution equation is then represented by the linear expression 
\begin{equation}
    k_i^{j+1} = k_i^j - r_{ij}^j,
\end{equation}
for all $i > j$, and infinity otherwise.

The last remaining piece to incorporate is the $S_{xy}$ schedulability constraints~\eqref{eq:xyschedulability} at each iteration of the $IS$ algorithm. Using the familiar superscript, we denote $x^j$ and $y^j$ as the integers $x$ and $y$ respectively at iteration $j$. Furthermore, we denote their binary expansions as $\boldsymbol{\chi^j}$ and $\boldsymbol{\psi^j}$ respectively. Recall that at iteration $j$, $k_i^j$ must be a power of $2$ multiple of either $x^j$ or $y^j$ for all $i \geq j$ (all prior elements have been removed from the vector at this iteration). If $k_i^j$ is a multiple of $x^j$, we say it is normalized with respect to $x^j$, and otherwise that it is normalized with respect to $y^j$.

Denote the vector of values in $\boldsymbol{k^j}$ that are normalized with respect to $x^j$ as $\boldsymbol{k_x^j}$, and let the density of this vector be $\rho_x^j$. We define the equivalent quantities $\boldsymbol{k_y^j}$ and $\rho_y^j$ for values normalized with respect to $y^j$. Then let 
\begin{align}
\begin{aligned}\label{eq:xyceiling}
    &\beta_x^j \geq x^j \rho_x^j = \sum_{l=0}^{\lfloor \log k_{\max} \rfloor} 2^l \chi_l^j \rho_x^j, \ \beta_x^j \in \mathbb{Z}, \\
    &\beta_y^j \geq y^j \rho_y^j = \sum_{l=0}^{\lfloor \log k_{\max} \rfloor} 2^l \psi_l^j \rho_y^j, \ \beta_y^j \in \mathbb{Z},
\end{aligned}
\end{align}
which has the familiar form of the sum of bilinear terms with one binary factor.

Let $\delta_{xl}^j \triangleq \chi_l^j \rho_x^j$ and similarly $\delta_{yl}^j \triangleq \psi_l^j \rho_y^j$. Applying the envelope technique,
\begin{align}
\begin{aligned}
    &\delta_{xl}^j \geq 0, \ \delta_{yl}^j \geq 0, \\
    &\delta_{xl}^j \leq \rho_x^j, \ \delta_{yl}^j \leq \rho_y^j, \\
    &\delta_{xl}^j \leq \chi_l^j, \ \delta_{yl}^j \leq \psi_l^j, \\
    &\delta_{yl}^j \geq \rho_x^j - (1 - \chi_l^j), \ \delta_{yl}^j \geq \rho_y^j - (1 - \psi_l^j),
\end{aligned}
\end{align}
and the bilinear terms in~\eqref{eq:xyceiling} can be replaced with $\delta_{xl}^j$ and $\delta_{yl}^j$ respectively.

Next define real numbers $\alpha_x^j$ and $\alpha_y^j$ such that 
\begin{align}
\begin{aligned}\label{eq:xyratio}
    &\beta_x^j = x^j \alpha_x^j = \sum_{l=0}^{\lfloor \log k_{\max} \rfloor} 2^l \chi_l^j \alpha_x^j, \\
    &\beta_y^j = y^j \alpha_y^j = \sum_{l=0}^{\lfloor \log k_{\max} \rfloor} 2^l \psi_l^j \alpha_y^j, \\
\end{aligned}
\end{align}
along with values $\omega_{xl}^j \triangleq \chi_l^j \alpha_x^j$ and $\omega_{yl}^j \triangleq \psi_l^j \alpha_y^j$. This is equivalent to 
\begin{align}
\begin{aligned}
    &\omega_{xl}^j \geq 0, \ \omega_{yl}^j \geq 0, \\
    &\omega_{xl}^j \leq \alpha_x^j, \ \omega_{yl}^j \leq \alpha_y^j, \\
    &\omega_{xl}^j \leq \chi_l^j, \ \omega_{yl}^j \leq \psi_l^j, \\
    &\omega_{xl}^j \geq \alpha_x^j - (1 - \chi_l^j), \ \omega_{yl}^j \geq \alpha_y^j - (1 - \psi_l^j),
\end{aligned}
\end{align}
and once again the bilinear terms in~\eqref{eq:xyratio} can be replaced with these values.

Recall that the $S_{xy}$ schedulability constraint~\eqref{eq:xyschedulability} is 
\begin{equation}
    \frac{\lceil x \rho(\boldsymbol{k'_x}) \rceil}{x} + \frac{\lceil y \rho(\boldsymbol{k'_y}) \rceil}{y} \leq 1.
\end{equation}
Examining the variables defined above, one can see that $\beta_x^j$ and $\beta_y^j$ are equivalent to the numerators on the left hand side, and $\alpha_x^j$ and $\alpha_y^j$ are equivalent to the full terms on the left hand side. Therefore, if $\alpha_x^j + \alpha_y^j \leq 1$, then the vector is schedulable. We only require one iteration of $IS$ to be schedulable, and we introduce the binary vector $\boldsymbol{\theta}$ to indicate whether each iteration is schedulable or not. This can be expressed as
\begin{align}
\begin{aligned}\label{eq:iterschedulability}
    &(\alpha_x^j + \alpha_y^j) \theta_j \leq 1, \ \forall \ 0 \leq j < M, \\
    &\boldsymbol{\theta}^T \mathbf{1}_M \geq 1. 
\end{aligned}
\end{align}
By the usual process, we define $\iota_x^j \triangleq \alpha_x^j \theta_j$ and $\iota_y^j \triangleq \alpha_y^j \theta_j$, or equivalently,
\begin{align}
\begin{aligned}
    &\iota_x^j \geq 0, \ \iota_y^j \geq 0, \\
    &\iota_x^j \leq \alpha_x^j, \ \iota_y^j \leq \alpha_y^j, \\
    &\iota_x^j \leq \theta_j, \ \iota_y^j \leq \theta_j, \\
    &\iota_x^j \geq \alpha_x^j - (1 - \theta_j), \ \iota_y^j \geq \alpha_y^j - (1 - \theta_j),
\end{aligned}
\end{align}
and replace the bilinear terms in~\eqref{eq:iterschedulability}.

The next step is to determine which values $k_i^j$ are normalized with respect to $x^j$ and $y^j$, and to enforce that each is a power of $2$ multiple of the corresponding value. Let $\nu_i^j = 1$ if $k_i^j$ is normalized with respect to $x^j$ and $0$ otherwise. Then 
\begin{equation}\label{eq:originalk}
    k_i^j = \sum_{n=0}^{\lfloor \log k_{\max} \rfloor} 2^n \eta_{in}^j (\nu_i^j x_i^j + (1-\nu_i^j) y_i^j),
\end{equation}
where $\eta_{in}^j = 1$ if the power of $2$ factor is $2^n$. First define $p_{il}^j \triangleq \nu_i \chi_{il}^j$ and $q_{il}^j \triangleq \nu_i \psi_{il}^j$, or equivalently
\begin{align}
\begin{aligned}
    &p_{il}^j \geq 0, \ q_{il}^j \geq 0, \\
    &p_{il}^j \leq \nu_i, \ q_{il}^j \leq \nu_i, \\
    &p_{il}^j \leq \chi_{il}^j, \ q_{il}^j \leq \psi_{il}^j, \\
    &p_{il}^j \geq \nu_i + \chi_{il}^j - 1, \ q_{il}^j \geq \nu_i + \psi_{il}^j - 1.
\end{aligned}
\end{align}
Then~\eqref{eq:originalk} becomes
\begin{align} 
\begin{aligned}\label{eq:kexpansion}
    k_i^j &= \sum_{n=0}^{\lfloor \log k_{\max} \rfloor} 2^n \eta_{in}^j \Big( \sum_{l=0}^{\lfloor \log k_{\max} \rfloor} 2^l (p_{il}^j - q_{il}^j + \psi_{il}^j) \Big) \\
    &= \sum_{n=0}^{\lfloor \log k_{\max} \rfloor} 2^n \eta_{in}^j \Big( \sum_{l=0}^{\lfloor \log k_{\max} \rfloor} 2^l h_{il}^j \Big) \\
    &= \sum_{l=0}^{\lfloor \log k_{\max} \rfloor} \sum_{n=0}^{\lfloor \log k_{\max} \rfloor} 2^{l+n} \eta_{in}^j h_{il}^j,
\end{aligned}
\end{align}
where $h_{il}^j \triangleq p_{il}^j - q_{il}^j + \psi_{il}^j$. 

Let $\phi_{iln}^j \triangleq \eta_{in}^j h_{il}^j$, which is again equivalent to 
\begin{align}
\begin{aligned}
    &\phi_{iln}^j \geq 0, \\
    &\phi_{iln}^j \leq \eta_{in}^j, \\
    &\phi_{iln}^j \leq h_{il}^j, \\
    &\phi_{iln}^j \geq \eta_{in}^j + h_{il}^j - 1.
\end{aligned}
\end{align}
Replacing the bilinear term in~\eqref{eq:kexpansion} with $\phi_{iln}^j$, we finally arrive at a linear expression for $k_i^j$. To ensure that $\eta_{in}^j = 1$ for only one value of $n$, we also have
\begin{equation}
    \boldsymbol{\eta_i^j}^T \mathbf{1}_{\lfloor \log k_{\max} \rfloor} = 1.
\end{equation}

Finally, we address the last step, which is computing the densities $\rho_x^j$ and $\rho_y^j$. First define $\mu_i^j$ such that $\mu_i^j k_i^j = 1$, which is equivalent to
\begin{equation}\label{eq:mukappa}
    \sum_{l=0}^{\lfloor \log k_{\max} \rfloor} \mu_i^j \kappa_{il}^j = 1.
\end{equation}
Next define $u_{il}^j \triangleq \mu_i^j \kappa_{il}^j$ and $v_i^j \triangleq \mu_i^j \nu_i^j$, or equivalently,
\begin{align}
\begin{aligned}
    &u_{il}^j \geq 0, \ v_i^j \geq 0, \\
    &u_{il}^j \leq \mu_i^j, \ v_i^j \leq \mu_i^j, \\
    &u_{il}^j \leq \kappa_{il}^j, \ v_i^j \leq \nu_i^j, \\
    &u_{il}^j \geq \mu_i^j - (1 - \kappa_{il}^j), \ v_i^j \geq \mu_i^j - (1 - \nu_i^j).
\end{aligned}
\end{align}
Then we can replace the bilinear term in~\eqref{eq:mukappa} with $u_{il}^j$, and observe that
\begin{align}
\begin{aligned}
    &\rho_x^j = \sum_{i=0}^{M-1} \nu_i^j \mu_i^j = \sum_{i=0}^{M-1} v_i^j, \\
    &\rho_y^j = \sum_{i=0}^{M-1} (1 - \nu_i^j) \mu_i^j = \sum_{i=0}^{M-1} \mu_i^j - v_i^j,
\end{aligned}
\end{align}
which completes the process.

\bibliographystyle{IEEEtran}
\bibliography{tree_num}

\end{document}